\documentclass[onecolumn]{IEEEtran}

\usepackage{hyperref}
\usepackage{times}
\usepackage{bm,mathptmx,amsfonts,amssymb,amstext,textcomp,amsmath}
\usepackage{times,verbatim,multirow,epsfig,graphics,graphicx,epstopdf}
\usepackage{url}

\usepackage{tikz}

\newcommand{\rea}{\mathbb{R}}

\newtheorem{mythm}{\bf Theorem}
\newtheorem{myass}{\bf Assumption}
\newtheorem{mylem}{\bf Lemma}
\newtheorem{myrem}{\bf Remark}

\DeclareMathAlphabet\mathbfcal{OMS}{cmsy}{b}{n}

\renewcommand{\baselinestretch}{1.5}

\title{\LARGE \bf
Leader-follower synchronization and ISS analysis for a
network of boundary-controlled wave PDEs}
\author{Luis Aguilar, Yury~Orlov, Alessandro~Pisano \\ $\;$\\
\textbf{\emph{This document is an enhanced version of a companion paper, currently under review for journal} } \\ \textbf{ \emph{publication, containing more detailed proofs of the main Theorems 1 and 2}.}
\thanks{This work has been supported by Mexican CONACYT grants 285279 and A1-S-9270.}
\thanks{A. Pisano ({\tt\small apisano@unica.it}) is with the Department
of Electrical and Electronic Engineering (DIEE), University of Cagliari, Cagliari,
09123, Italy. Y. Orlov ({\tt\small yorlov@cicese.mx}) is with CICESE Research Center, Electronics and Telecommunication Department, Ensenada, Mexico. L. Aguilar ({\tt \small laguilarb@ipn.mx}) is with Instituto Polit\'ecnico Nacional, Ave. Instituto Polit\'ecnico Nacional 1310 Col. Nueva Tijuana, Tijuana, 22435, M\'exico.}
}

\begin{document}

\maketitle

\begin{abstract}
A network of agents, modeled by a class of wave PDEs, is under investigation. One agent in the network plays the role of a leader, and all the remaining $\lq\lq$follower" agents are required to asymptotically track the state of the leader. Only boundary sensing of the agent's state is assumed, and each agent is controlled through the boundary by Neumann-type actuation. 
A linear interaction protocol is proposed and analyzed by means of a Lyapunov-based approach. A simple set of tuning rules, guaranteeing the exponential achievement of synchronization, is obtained. In addition, an exponential ISS relation, characterizing the effects on the tracking accuracy of {boundary and in-domain} disturbances, is derived for the closed loop system.
\end{abstract}


\section{Introduction}\label{sect1}

The \emph{consensus} problem seeks to enforce agreement amongst the states of networked dynamical systems by penalizing their local disagreement with the neighboring nodes in a dynamic manner. 
A particular class of consensus problems is the leader-follower decentralized tracking, where a specific agent in the network plays the role of a leader and all remaining follower agents aim to synchronize their state evolutions  with that of the leader (see e.g. \cite{cao2012distributed}).

It is worth noting that the consensus problem for networks of distributed parameter systems has not received yet the same level of attention as its finite-dimensional counterpart.

In \cite{li2014exact,li2018exact}, exact synchronization for a set of coupled wave processes, {part of which} equipped by a boundary control input, was provided in the two cases of Dirichlet and Neumann actuation. In \cite{ConsensusWaveNeurocomput2018}, the leaderless consensus problem was addressed with reference to  multi agent systems where agents dynamics are governed by heat and wave dynamics with distributed control. In \cite{HeIET2018}, leader follower consensus for perturbed parabolic PDEs with distributed control was achieved by means of an adaptive unit-vector sliding mode controller. In \cite{demetriou2013synchronization}, the consensus problem for a network of agents modeled by a class of parabolic PDEs, and communicating through undirected communication topologies, has been studied. In \cite{PilloConsTACHeat,PilloConsCPDELeaderFollowHeat2016}  the leaderless and leader-following consensus problems for perturbed diffusion PDEs were solved through sliding-mode based boundary control. In \cite{BipartiteSCL2020}, the leader-less consensus problem was dealt with for a multi agent system where agents dynamics are governed by a class of perturbed boundary controlled wave processes. Motivated by the above state-of-art analysis, and with the aim of developing a leader-following consensus controller for networked wave PDEs,  a "pointwise-in-space" agreement between the follower and leader profiles,  is established in the present paper via Lyapunov analysis by using a linear PD-like local interaction rule. In contrast to the related investigation  \cite{BipartiteSCL2020} where the leader-less consensus was studied within the same wave PDEs framework, we focus on the leader-following  consensus case.  
An ISS analysis is also made to investigate the effect of {boundary and in-domain} disturbances on the closed-loop system accuracy. Certain tuning inequalities, which  are more restrictive, compared to those derived in the unperturbed scenario, are to be imposed on the controller parameters in order to ensure an exponential ISS inequality. The contribution to the existing literature is thus as follows.

i. The leader-following consensus problem is addressed and solved with reference to multi agent systems with agents' and leader dynamics, governed by the wave equation with {Neumann}-type boundary control.

ii. The proposed local interaction rule ensures the pointwise convergence to zero of the deviation between the leader and follower trajectories.


iii. The effects of boundary and in-domain disturbances on the consensus accuracy are constructively analyzed from the ISS standpoint.

The rest of the paper is outlined as follows. In Section~\ref{sect2} some mathematical preliminaries on graph theory and useful norm properties and definitions are recalled. The communication protocol providing consensus-tracking is studied in  Section~\ref{sect3}. The ISS analysis of the closed loop system in the presence of  {boundary and distributed in-domain} disturbances is made in Section \ref{sectDIST}. Simulation results are presented in Section~\ref{sect4}, and conclusions and perspectives for next investigations are collected in Section~\ref{sect5}.

\section{Mathematical Preliminaries and Notations}\label{sect2}

\subsection{Useful definitions and properties}
The Euclidean norm of the real-valued $n$-dimensional vector ${x} = [x_1,\dots,x_n]^T \in \rea^n$ is defined as
$\|{x}\|_{{2}}= \left( \Sigma_{i=1}^n |x_i|^\mathrm{2}\right)^{1/\mathrm{2}}\equiv\sqrt{{x}^T{x}}$.  The next well-known inequalities to hold true for all $x, y \in \rea^n$ and for arbitrary $\xi>0$, are recalled:
\begin{equation}{\small
\left|{x}^T{y}\right|\leq\|{x}\|_{{2}}\|{y}\|_{2} \leq \frac{\xi }{{2}}\|{x}\|_{{2}}^{2}+\frac{1}{{2\xi}}\|{y}\|_{{2}}^{2}, \;\;\;\; \xi > 0.}
\label{app1:HolderInequality}
\end{equation}
Given a symmetric positive definite matrix $M\in\rea^{n\times n}$, let $\lambda_m(M)$ and $\lambda_M(M)$ denote the minimum and maximum eigenvalues of $M$. 
The symbols ${ 1}_n=[1,1,\dots,1]^T\in\rea^n$ and ${ 0}_n=[0,0,\dots,0]^T\in\rea^n$ stand for the all-ones and all-zeros vectors. Let $x(t):\rea^+ \cup \left\{0 \right\} \rightarrow \Re$ be a scalar function. Then the notation
\begin{equation}\label{essdef}
E_s\left(x(t)\right) \triangleq ess \; \sup_{\tau\in [0,t]} x(\tau)
\end{equation}
is used for brevity.
\vspace{0.1cm}
{$L_2$ stands for the Hilbert space of square integrable scalar functions $z(\varsigma)$ on the domain $(0,1)$ with the corresponding $L_2$-norm
\begin{eqnarray}\label{define2}
  \|z(\cdot)\|_{L_2} = \sqrt{\int_0^1 z^{2} (\varsigma) d \varsigma}.
\end{eqnarray}}
{The symbol $L_\infty(0,T;L_2)$ is reserved for the set of functions $f(\varsigma,t)$ such that $f(\cdot,t)\in L_2$ for almost all $t\in (0,T),$  $\int_0^1 f(\varsigma,t)\phi(\varsigma) d\varsigma $ is Lebesgue measurable in $t$ for all $\phi(\cdot)\in L_2$, and $E_s\left( \int_0^1f^2(\varsigma,t) d\varsigma \right)<\infty$. It is said that  $f(\cdot)\in L^{loc}_\infty(L_2(a,b))$ iff $f(\cdot)\in L_\infty(0,T;L_2(a,b))$ for all $T>0$.}

$\mathrm{H}^{\ell}$, with $\ell=\textcolor{blue}{1},2,\dots$, denotes the Sobolev space of absolutely continuous scalar functions $w(\varsigma)$ on the domain $(0,1)$, with square integrable derivatives $w^{(k)}(\varsigma)$ up to order $\ell$ and the $\mathrm{H}^\ell$-norm
\begin{equation}
\begin{array}{c}
\|{z}(\cdot)\|_{\mathrm{H}^{\ell}}
=\sqrt{\int_0^1 \sum_{k=0}^\ell \left[{w}^{(k)}(\varsigma)\right]^2 ~d\varsigma}.
\end{array}\label{app1:L2normScalare}
\end{equation}
The first and second order derivatives ${z}^{(1)}(\varsigma)$ and  ${z}^{(2)}(\varsigma)$ will also be denoted as   ${z}^{(1)}(\varsigma)=z_\varsigma(\varsigma)$ and ${z}^{(2)}(\varsigma)=z_{\varsigma \varsigma}(\varsigma)$. {In addition, the notations
$$ L^n_2 =
\begin{array}{c}
                         \underbrace{L_2(0,1) \times L_2(0,1) \times \ldots \times L_2(0,1)} \\
                         n \;\; \texttt{times}
                       \end{array},$$}
$$H^{\ell,n} =
\begin{array}{c}
                         \underbrace{H^\ell(0,1) \times H^\ell(0,1) \times \ldots \times H^\ell(0,1)} \\
                         n \;\; \texttt{times}
                       \end{array}$$
are utilized and
\begin{eqnarray}\label{app1:L2normVector}
 {\|z(\cdot)\|_{L_2^n} = \sqrt{\sum_{i=1}^{n}   \|z_i(\cdot)\|^2_{L_2}},}   \;\;\;\;\;   \|w(\cdot)\|_{H^{\ell,n}} = \sqrt{\sum_{i=1}^{n} \|w_i(\cdot)\|^2_{H^\ell}}
\end{eqnarray}
stand, {respectively, for the $L_2$-norm of a vector function
$z(\varsigma)=[z_1(\varsigma), z_2(\varsigma), ...., z_n(\varsigma)] \in
L_2^n$ and} for the $H^\ell$-norm of a vector function $w(\varsigma)=[w_1(\varsigma), w_2(\varsigma), ...., w_n(\varsigma)] \in
H^{\ell,n}$.

The following well-known Lemma constitutes a vector counterpart of the Poincare inequality
\begin{mylem}\label{lemma1:pisanoResult}{\it
Let ${b}(\varsigma)\in\mathrm{H}^{1,n}$. Then, the following inequality holds:
\begin{equation}{\small
\begin{array}{c}
\|{b}(\cdot)\|_{{L_2^n}}^2\leq 2\left(\|{b}(i)\|_{{2}}^2+\|{b}_\varsigma(\cdot)\|_{{L_2^n}}^2\right),\;\; i=0,1.
\end{array}}\label{app1:pisanoResult}
\end{equation}}
\end{mylem}

\vspace{0.2cm}
{The specific formulation \eqref{app1:pisanoResult} of the Poincare inequality can be found in \cite{PilloConsTACHeat}.}

\subsection{Algebraic Graph Theory definitions and properties}

Consider a group of $n$ dynamical agents along with the undirected
graph $\mathbfcal{G}(\mathbfcal{V},\mathbfcal{E},\mathbfcal{A})$ modeling the communication topology among these systems, where $\mathbfcal{V}=\{1,\ldots,n\}$ is the node set and $\mathbfcal{E}\subseteq \{\mathbfcal{V} \times \mathbfcal{V}\}$ is the edge set. An edge $(i,j)\in \mathbfcal{E}$ if agents $i$ and $j$ can exchange information. The adjacency matrix $\mathbfcal{A}=[a_{ij}]$ associated with $\mathbfcal{G}$ is such that {$a_{ii}=0$,} $a_{ij}=1$  if $(j,i)\in \mathbfcal{E}$, and $a_{ij}=0$ otherwise. 
A path in an undirected graph $\mathbfcal{G}$ is a sequence of edges joining two nodes of the graph. An undirected graph is said to be \emph{connected} if there is a path between every pair of nodes. Throughout, the \emph{Laplacian Matrix} $\mathbfcal{L}=[\ell_{ij}]\in\rea^{n\times n}$, associated with the graph $\mathbfcal{G}$, is defined as $\ell_{ii}=\Sigma_{j=1,i\neq j}^n a_{ij}$ and $\ell_{ij}=-a_{ij}$, $i\neq j$. 

\section{Distributed coordinated tracking for Networked Wave Processes}\label{sect3}

Consider a set of $n$ dynamical agents $\mathbfcal{V}_f=\lbrace1,2,\dots,n\rbrace$, identified as \emph{followers}, which are governed by the wave equation, expressed in the  vector form
\begin{flalign}
u_{tt}(\varsigma,t)&=u_{\varsigma\varsigma}(\varsigma,t)
\label{sect1:barDynamic}\\
u(\varsigma,0)&= u^0(\varsigma), \;\;\; u_t(\varsigma,0)= u_t^0(\varsigma),
\label{sect1:barICs}\\
u_\varsigma(0,t)&=c_0 u_t(0,t)\label{sect1:freeend}\\
u_\varsigma(1,t)&=q(t).\label{sect1:actuation} 
\end{flalign}
Hereinafter, $u(\varsigma,t)=[u_1(\varsigma,t),u_2(\varsigma,t),\dots,u_n(\varsigma,t)]^T$ is the vector, collecting the states of all followers ($u_i(\varsigma,t)$ denotes the transverse displacement of the i-th agent at position $\varsigma\in(0,1)$ and time $t \geq 0$),  $q(t)=[q_1(t), q_2(t), \ldots, q_n(t)]^T$ is the vector, collecting the agents' Neumann-type boundary control inputs and $c_0$ a positive constant. Follower agents are supposed to be communicating each other through a static, undirected topology described by $\mathbfcal{G}_f(\mathbfcal{V}_f,\mathbfcal{E},\mathbfcal{A}_f)$.

{\begin{myrem}\label{remarkstability}
{\it 
By using the Lyapunov functional $W(t)=\frac{1}{2}\|\tilde{u}_\varsigma(\cdot,t)\|_{{L^n_2}}^2+\frac{1}{2}\|\tilde{u}_t(\cdot,t)\|_{{L^n_2}}^2$, 
whose time derivative along the solutions of the open-loop system \eqref{sect1:barDynamic}-\eqref{sect1:actuation} with $q(t)=0$ is $\dot W(t)=- c_0 \|\tilde{u}_t(0,t) \|_2^2$, it is concluded that system  \eqref{sect1:barDynamic}-\eqref{sect1:actuation}  is stable in the open-loop but not asymptotically stable, thus motivating the need for consensus-based control to achieve synchronization between agents. 
$\square$}\end{myrem}}
\vspace{0.1cm}

Suppose that in addition to the $n$ followers there exists a leader agent, labeled with the index number $0$ and governed by the unforced boundary-value problem
\begin{flalign}
 u_{0,tt}(\varsigma,t)&=u_{0,\varsigma\varsigma}(\varsigma,t), \label{sect1:leader} \\
 u_0(\varsigma,0) &= u_0^0(\varsigma), \;\; u_{0,t}(\varsigma,0) = u_{0,t}^0(\varsigma), \label{sect1:leaderIC}\\
 u_{0,\varsigma}(0,t)&=c_0 u_{0,t}(0,t),  \label{sect1:leaderBC0} \\
 u_{0,\varsigma}(1,t)&={0}.\label{sect1:leaderBC1} 
\end{flalign}

It is assumed that the leader's boundary information $(u_{0}(1,t), u_{0,t}(1,t))$ is available to a nonempty subset of followers. 
Let $a_{i0}=1$ if the leader communicates with the $i$-th follower ($i=1,2,\dots,n$), and $a_{i0}=0$ otherwise.


\begin{myass}\label{assumption:graph}
{\it Follower agents exchange information according to the static and undirected topology $\mathbfcal{G}_f(\mathbfcal{V}_f,\mathbfcal{E},\mathbfcal{A}_f)$ that is assumed to be connected, and the leader communicates with at least one follower.}
\end{myass}

The following regularity and compatibility conditions are in force {to deal with classical solutions of the boundary value problem in question (see \cite{curtain1995introduction} for details).}

\begin{myass}\label{assumption:1_distrurbance}
{\it
The ICs in \eqref{sect1:barICs} and \eqref{sect1:leaderIC} are such that
{\begin{equation}
\begin{array}{c}
  u^0(\varsigma),\ u_t^0(\varsigma)\in \mathrm{H}^{2,n}, \;\;\;
  u_0^0(\varsigma),\  u_{0,t}^0(\varsigma)\in \mathrm{H}^{2}
\end{array}
\label{sect1:regul}
\end{equation}}
and the following compatibility conditions hold
\begin{equation}
\begin{array}{c}
  u_\varsigma^0(0)=c_0 u_t^0(0), \;\;\;   u_\varsigma^0(1)=q(t) \\
  u_{0,\varsigma}^0(0)=c_0 u_{0,t}^0(0), \;\;\;   u_{0,\varsigma}^0(1)={0_n} \\
\end{array}
\label{sect1:compa}
\end{equation}
}\end{myass}

\vspace{.2cm}
{With the assumption above, the stability of the collective networks dynamics of the leader $(u_0(\cdot,t), u_{0,t}(\cdot,t))$ and followers $(u(\cdot,t),  u_{t}(\cdot,t))$ is studied in a proper Sobolev space being specified to  $\mathrm{H}^{1,n+1}\times L_2^{n+1}$.}


\subsection{Problem Statement}\label{sect3_1}

In the sequel, we design the agents' control inputs $q_i(t)$ such that all followers states $u_i(\varsigma,t)$ ($i=1,2,...,n$) asymptotically track the leader's state  $u_0(\varsigma,t)$. {Each agent communicates continuously to its neighbours by accessing to their local boundary measurements $u_i(1,t)$ and $u_{i,t}(1,t)$. Note that the leader agent is a source node of the overall directed graph including both leader and followers, and thus it does not receive any data and only sends its own boundary measurements to its neighbours. }
The control task is specifically to enforce the {exponential point-wise consensus relation
\begin{equation}
 \max_{x\in[0,1]}|u_i(\varsigma,t)-u_0(\varsigma,t)|^2  \leq \delta e^{-\alpha t},\;\;\;\;\forall~i\in \mathbfcal{V}_f
\label{sect1:consensus_tracking}
\end{equation}
for some nonnegative constants $\delta$ and $\alpha$.} {Inspired by the natural consensus algorithm for a network of double integrators (see e.g. \cite{renijrnc07natural}),}  we propose the local interaction protocol
\begin{equation}
q(t)= -k_1 \mathbfcal{M} \tilde{u}(1,t) - k_2   \mathbfcal{M} \tilde{u}_t(1,t),
\label{sect1:consensusProthhocol0}
\end{equation}
where $k_1, k_2$ are nonnegative tuning constants, 
\begin{flalign}
\tilde{u}(\varsigma,t)=&u(\varsigma,t)-1_n u_0(\varsigma,t)\label{sect2:u_tilde}
\end{flalign}
represents the deviation vector of the followers' states with respect to the leader profile, and $$\mathbfcal{M}=\mathbfcal{L}_f+diag\lbrace a_{10},a_{20},\dots,a_{n0}\rbrace\in\rea^{n\times n}$$. Notice that under Assumption \ref{assumption:graph} matrix $\mathbfcal{M}$ turns out to be symmetric and positive definite, see e.g.  \cite{cao2012distributed} for details.

{\begin{myrem}\label{remarkleaderstability}
{\it 
By using the Lyapunov functional $W_0(t)=\frac{1}{2}\|{u}_{0,\varsigma}(\cdot,t)\|_{{L_2}}^2+\frac{1}{2}\|{u}_{0,t}(\cdot,t)\|_{{L_2}}^2$,
whose time derivative along the solutions of \eqref{sect1:leader}-\eqref{sect1:leaderBC1} is $\dot W_0(t)=- c_0 \tilde{u}^2_{0,t}(0,t)$, one concludes that leader's dynamics are stable but not asymptotically stable.
It is seen, in particular, that system \eqref{sect1:leader}-\eqref{sect1:leaderBC1} possesses a multitude of constant stable equilibria of the type $u_0(\varsigma,t)=U_0=const$, from which it derives that the system possesses a zero eigenvalue, associated with a constant (in space) eigenfunction. All follower agents will thus eventually synchronize to a constant profile determined by the leader's initial conditions. It is worth noticing that autonomous leader's dynamics are often considered in the literature, see e.g. \cite{JadbaMOrseTAC03,LiRenTAC2014}. 
Considering non-autonomous leader's dynamics generally requires that all follower agents must know the leader's input signal, as in \cite{SONGSCL2009,ZhouSCL2012}, thereby compromising the distributed nature of the local interaction protocol. 
Only more recently (see, e.g., \cite{cao2012distributed}) this restriction 
has been successfully removed by means of nonlinear discontinuous local interaction laws capable of "rejecting" the matching disturbance represented by the exogenous leader's driving signal. 
This challenging task however remains beyond the scope of the paper which is the first research work, addressing the leader-following consensus problem for agents' dynamics, governed by the wave PDE. The above challenge will be pursued in our future research.
$\square$}\end{myrem}}
\vspace{0.1cm}

\subsection{Convergence analysis}\label{sect3_2}

The performance of the collective agents' dynamics under the local interaction protocol \eqref{sect1:consensusProthhocol0}
is going to be investigated and simple tuning rules are constructively derived such that condition \eqref{sect1:consensus_tracking} is guaranteed. The boundary value problem governing the deviation variable $\tilde{u}(\varsigma,t)$ then reads as
\begin{flalign}
&\tilde{u}_{tt}(\varsigma,t)=\tilde{u}_{\varsigma\varsigma}(\varsigma,t), \label{sect2:errDyn}\\
&\tilde{u}(\varsigma,0)={u}(\varsigma,0)-1_n u_0(\varsigma,0),\\
&\tilde{u}_t(\varsigma,0)={u}_{t}(\varsigma,0)-1_n u_{0,t}(\varsigma,0),\label{sect2:errIC}\\
&\tilde{u}_{\varsigma}(0,t)=c_0 \tilde{u}_{t}(0,t),\label{sect2:errBC0} \\
&\tilde{u}_{\varsigma}(1,t)=q(t)=-k_1 \mathbfcal{M}\tilde{u}(1,t)-k_2 \mathbfcal{M}\tilde{u}_t(1,t).
\label{sect2:errBC}
\end{flalign}

The properties of the closed-loop agent's dynamics are investigated in Theorem \ref{lemma1:Vconv} by Lyapunov analysis considering the candidate functional
\begin{flalign}
V(t)=&E(t)+G_1(t)+G_2(t),  \label{VCA} \\
E(t)=& \frac{1}{2}\int_0^1 \tilde{u}_\varsigma(\zeta,t)^T \tilde{u}_\varsigma(\zeta,t)d\zeta +  \frac{1}{2}\int_0^1 \tilde{u}_t(\zeta,t)^T \tilde{u}_t(\zeta,t)d\zeta+ \frac{1}{2}k_1 \tilde{u}^T(1,t)\mathbfcal{M}\tilde{u}(1,t),\label{ECA} \\
G_1(t)=&   \frac{1}{2}\rho_1 k_2 \tilde{u}^T(1,t)\mathbfcal{M}\tilde{u}(1,t) + \rho_1 \int_0^1  \tilde{u}(1,t)^T \tilde{u}_t(\zeta,t)d\zeta,\label{G1CA} \\
G_2(t)=&  \rho_2 \int_0^1 (\zeta - 1) \tilde{u}_t(\zeta,t)^T \tilde{u}_\varsigma(\zeta,t)d\zeta, \label{G2CA}
\end{flalign}
with $\rho_1$, $\rho_2$,  positive constants to be defined. The quadratic functional
\begin{flalign}
  V_0(t) = \|\tilde{u}_\varsigma(\cdot,t)\|_{{L^n_2}}^2 + \|\tilde{u}_t(\cdot,t)\|_{{L_2^n}}^2 + \|\tilde{u}(1,t)\|_2^2,  \label{cond15CA}
\end{flalign}
relying on appropriate norms, and the constants
\begin{small}
\begin{equation}\label{tau1}
\tau_1 = \min{ \left\{   \left( \frac{ 1 - \rho_2 - \rho_1 }{2}\right)  , k_1\lambda_m(\mathbfcal{M}) + \rho_1 k_2 \lambda_m(\mathbfcal{M})-\rho_1  \right\}  },
\end{equation}
\begin{eqnarray}
\tau_2 &=  \max{ \left\{   \frac{1 + \rho_2 + \rho_1 }{2} , k_1\lambda_M(\mathbfcal{M}) + \rho_1, k_2 \lambda_M(\mathbfcal{M})+\rho_1  \right\}  }\label{tau} \\
\mu &= \min \left\{ \frac{1}{2}\rho_2 , \frac{1}{2}(\rho_2-\rho_1) ,  \rho_1 \left[k_1 \lambda_m(\mathbfcal{M})-\frac{1}{2}c_0\right]\right\},\label{mudeff}
\end{eqnarray}
\end{small}
will also be used throughout.  {Provided that
\begin{flalign}
k_1 > \frac{c_0}{2 \lambda_m(\mathbfcal{M})} , \;\;\; k_2 > 0,  \label{cond1lCA}
\end{flalign}
where the positive constant $c_0$ is the same as in the boundary conditions \eqref{sect1:freeend}, \eqref{sect1:leaderBC0},
and
\begin{small}
\begin{flalign}
0 &< \rho_1 < \min \left( k_1\lambda_m(\mathbfcal{M}), 2 k_2 \lambda_m(\mathbfcal{M}), 1-\rho_2, \rho_2, 2-\frac{c_0}{\rho_2(1+c_0^2)} \right) \label{ro1teo1}\\
0 &< \rho_2 <  \min \left( 1, \frac{2 c_0}{1+c_0^2}\right), \label{ro2teo1}
\end{flalign}
\end{small}
the next result is in force}.

\vspace{.1cm}

\begin{mythm}\label{lemma1:Vconv}
Consider the followers' and leader's dynamics \eqref{sect1:barDynamic}-\eqref{sect1:actuation} and \eqref{sect1:leader}-\eqref{sect1:leaderBC1}  along with the local interaction rule \eqref{sect1:consensusProthhocol0}, \eqref{cond1lCA}. {Let  Assumptions \ref{assumption:graph} and \ref{assumption:1_distrurbance}
hold true. Then, the over-all closed-loop system globally possesses a unique classical solution} and the point-wise consensus condition \eqref{sect1:consensus_tracking} {is satisfied with $\delta=   \frac{\left(1+\sqrt{2}\right)}{\tau_1} V(0) $ and $\alpha=\frac{\mu}{\tau_2}$, where constant $V(0)$ is computed by \eqref{VCA}-\eqref{G2CA} whereas $\mu,\tau_1$ and $\tau_2$ are given in \eqref{tau1}-\eqref{tau} with arbitrary coefficients $\rho_1$ and $\rho_2$} {subject to \eqref{ro1teo1}-\eqref{ro2teo1}.}
\end{mythm}
\begin{proof} {First let us note that there locally exists a unique classical solution of the closed-loop system \eqref{sect1:barDynamic}-\eqref{sect1:actuation},  \eqref{sect1:leader}-\eqref{sect1:leaderBC1}, \eqref{sect1:consensusProthhocol0}, \eqref{cond1lCA}. To reproduce this conclusion it suffices to follow the same line of reasoning used the proof of \cite[Theorem 1]{orlov2020}. } {Next let} us derive under which conditions on the  $\rho_1$, $\rho_2$ constants functional \eqref{VCA}-\eqref{G2CA} is positive definite. By means of \eqref{app1:HolderInequality} one derives the following estimations
 \begin{flalign}
\left|\rho_1 \int_0^1  \tilde{u}(1,t)^T \tilde{u}_t(\zeta,t)d\zeta \right|\leq \frac{1}{2}\rho_1 \| \tilde{u}(1,t)\|_2^2 + \frac{1}{2}\rho_1\|\tilde{u}_t(\cdot,t)\|_{{L_2^n}}^2    \label{VdotCbnd1A}
\end{flalign}
\begin{equation}\label{BndG1}
|G_2(t)| \leq  \frac{1}{2}\rho_2 \|\tilde{u}_t(\cdot,t)\|_{{L_2^n}}^2 + \frac{1}{2}\rho_2  \|\tilde{u}_\varsigma(\cdot,t)\|_{{L_2^n}}^2.
\end{equation}
By substituting \eqref{VdotCbnd1A}-\eqref{BndG1} into  \eqref{VCA}-\eqref{G2CA} this yields that
\begin{flalign}
V(t) \geq &  \frac{1}{2} \left( 1 - \rho_2 \right) \|\tilde{u}_\varsigma(\cdot,t)\|_{{L_2^n}}^2 + \frac{1}{2} \left(  1 - \rho_2 - \rho_1 \right) \|\tilde{u}_t(\cdot,t)\|_{{L_2^n}}^2 + \frac{1}{2} \left(  k_1\lambda_m(\mathbfcal{M})+ \rho_1 k_2 \lambda_m(\mathbfcal{M}) -\rho_1 \right) \|\tilde{u}(1,t)\|_{2}^2. \label{BndG2}
\end{flalign}
Provided that the next inequalities hold
\begin{flalign}
0 &< \rho_1 < \min \left( k_1\lambda_m(\mathbfcal{M}), 1-\rho_2 \right),  \;\;\;\; 0 < \rho_2 < 1, \label{cond1lCAZ9}
\end{flalign}
it is straightforwardly concluded by \eqref{BndG2} and \eqref{VCA}-\eqref{G2CA} that
\begin{flalign}
\tau_1 V_0(t) \leq V(t) \leq \tau_2 V_0(t),\label{cond14CA}
\end{flalign}
where $V_0(t)$ is defined in \eqref{cond15CA}, and  the positive constants $\tau_1$ and $\tau_2$ are defined in \eqref{tau1}-\eqref{tau}. 


Let us now compute the time derivative of $V(t)$ along the solutions of \eqref{sect2:errDyn}-\eqref{sect2:errBC}. Differentiating \eqref{ECA}, substituting \eqref{sect2:errDyn} and \eqref{sect2:errBC0}-\eqref{sect2:errBC} in the resulting expression, performing integration by parts and finally rearranging yield the chain of equalities
\begin{flalign}
\dot E(t)=& \int_0^1 \tilde{u}_\varsigma(\zeta,t)^T \tilde{u}_{\varsigma t}(\zeta,t)d\zeta + \int_0^1 \tilde{u}_t(\zeta,t)^T \tilde{u}_{t t}(\zeta,t)d\zeta  + k_1 \tilde{u}^T(1,t)\mathbfcal{M}\tilde{u}_t(1,t)   \nonumber \\ =&\int_0^1 \tilde{u}_\varsigma(\zeta,t)^T \tilde{u}_{\varsigma t}(\zeta,t)d\zeta + \int_0^1 \tilde{u}_t(\zeta,t)^T \tilde{u}_{\varsigma \varsigma}(\zeta,t)d\zeta + k_1 \tilde{u}^T(1,t)\mathbfcal{M}\tilde{u}_t(1,t)  \nonumber \\
=&\int_0^1 \tilde{u}_\varsigma(\zeta,t)^T \tilde{u}_{\varsigma t}(\zeta,t)d\zeta + \left. \tilde{u}_t(x,t)^T \tilde{u}_{\varsigma}(x,t) \right|_{x=0}^{x=1}  -\int_0^1 \tilde{u}_\varsigma(\zeta,t)^T \tilde{u}_{\varsigma t}(\zeta,t)d\zeta +  k_1  \tilde{u}_t^T(1,t)\mathbfcal{M}\tilde{u}(1,t) \nonumber \\
= & \; \tilde{u}_t(1,t)^T \tilde{u}_{\varsigma}(1,t)-\tilde{u}_t(0,t)^T \tilde{u}_{\varsigma}(0,t) +  k_1  \tilde{u}_t^T(1,t)\mathbfcal{M}\tilde{u}(1,t) \nonumber \\
=&- k_2\tilde{u}_t^T(1,t)\mathbfcal{M}\tilde{u}_t(1,t)- c_0 \|\tilde{u}_t(0,t) \|_2^2. \label{EdotCALONG}
\end{flalign}

Differentiating \eqref{G1CA} along  \eqref{sect2:errDyn}, \eqref{sect2:errBC0}-\eqref{sect2:errBC} one obtains
\begin{flalign}
\dot G_1(t) & =\rho_1 k_2 \tilde{u}_t(1,t)^T\mathbfcal{M}\tilde{u}(1,t)+ \rho_1 \int_0^1  \tilde{u}_t(1,t)^T \tilde{u}_t(\zeta,t)d\zeta +  \rho_1 \int_0^1  \tilde{u}(1,t)^T \tilde{u}_{tt}(\zeta,t)d\zeta \nonumber \\ & =\rho_1 k_2 \tilde{u}_t(1,t)^T\mathbfcal{M}\tilde{u}(1,t) + \rho_1 \int_0^1  \tilde{u}_t(1,t)^T \tilde{u}_t(\zeta,t)d\zeta +  \rho_1 \int_0^1  \tilde{u}(1,t)^T \tilde{u}_{\varsigma \varsigma}(\zeta,t)d\zeta. \label{G1dotinitCA}
\end{flalign}
Straightforward integration and the BCs \eqref{sect2:errBC0}-\eqref{sect2:errBC} yield
\begin{flalign}
 \rho_1 \int_0^1  \tilde{u}(1,t)^T \tilde{u}_{\varsigma \varsigma}(\zeta,t)d\zeta & =  \left. \rho_1  \tilde{u}(1,t)^T \tilde{u}_{\varsigma}(x,t) \right|_{x=0}^{x=1}\nonumber \\ & =\rho_1 \tilde{u}(1,t)^T \tilde{u}_{\varsigma}(1,t)   - \rho_1 \tilde{u}(1,t)^T \tilde{u}_{\varsigma }(0,t) \nonumber \\& = - \rho_1 k_1 \tilde{u}(1,t)^T \mathbfcal{M} \tilde{u}_{}(1,t)  -\rho_1 k_2 \tilde{u}(1,t)^T \mathbfcal{M} \tilde{u}_{t}(1,t) - \rho_1 c_0 \tilde{u}(1,t)^T \tilde{u}_{t}(0,t). \label{G1dotinitCAH}
\end{flalign}
Substituting \eqref{G1dotinitCAH} into \eqref{G1dotinitCA} yields
\begin{flalign}
\dot G_1(t)  = & \rho_1 \int_0^1  \tilde{u}_t(1,t)^T \tilde{u}_t(\zeta,t)d\zeta - \rho_1 k_1 \tilde{u}(1,t)^T \mathbfcal{M} \tilde{u}_{}(1,t) -\rho_1 c_0 \tilde{u}(1,t)^T \tilde{u}_{t}(0,t). \label{G1dotinDFitCA}
\end{flalign}

Differentiating \eqref{G2CA}  yields
\begin{flalign}
\dot G_2(t) =&  
\rho_2  \int_0^1 (\zeta-1) \tilde{u}_{\varsigma \varsigma}(\zeta,t)^T \tilde{u}_\varsigma(\zeta,t)d\zeta  +  \rho_2  \int_0^1 (\zeta-1) \tilde{u}_{t}(\zeta,t)^T \tilde{u}_{\varsigma t}(\zeta,t)d\zeta.   \label{G2dotinitCApreli}
\end{flalign}
Integrating by parts and substituting \eqref{sect2:errDyn}, \eqref{sect2:errBC0}-\eqref{sect2:errBC} one derives
\begin{flalign}
\dot G_2(t) =&\left. \frac{\rho_2}{2}(x-1)\tilde{u}_\varsigma(x,t)^T\tilde{u}_\varsigma(x,t) \right|_{x=0}^{x=1}- \frac{1}{2} \rho_2 \|\tilde{u}_\varsigma(\cdot,t)\|_{{L_2^n}}^2  \nonumber \\
 & +  \left. \frac{\rho_2}{2}(x-1)\tilde{u}_t(x,t)^T\tilde{u}_t(x,t) \right|_{x=0}^{x=1}- \frac{1}{2} \rho_2 \|\tilde{u}_t(\cdot,t)\|_{{L_2^n}}^2 \nonumber \\
  =& \; \frac{\rho_2}{2}\tilde{u}_\varsigma(0,t)^T\tilde{u}_\varsigma(0,t) - \frac{1}{2} \rho_2 \|\tilde{u}_\varsigma(\cdot,t)\|_{{L_2^n}}^2  +   \frac{\rho_2}{2} \tilde{u}_t(0,t)^T\tilde{u}_t(0 ,t)  - \frac{1}{2} \rho_2 \|\tilde{u}_t(\cdot,t)\|_{{L_2^n}}^2 \nonumber \\
 =& -\frac{1}{2} \rho_2 \ \|\tilde{u}_\varsigma(\cdot,t)\|_{{L_2^n}}^2 -\frac{1}{2} \rho_2 \ \|\tilde{u}_t(\cdot,t)\|_{{L_2^n}}^2  + \frac{1}{2}\rho_2 (1+c_0^2) \tilde{u}_t(0,t)^T\tilde{u}_t(0,t). \label{G2dotinitCA}
\end{flalign}

Differentiating \eqref{VCA} and combining \eqref{EdotCALONG}, \eqref{G1dotinDFitCA} and \eqref{G2dotinitCA} one obtains
\begin{flalign}
\dot V(t) =&  \dot E(t)+\dot G_1(t)+\dot G_2(t) \nonumber \\ =&  - k_2\tilde{u}_t^T(1,t)\mathbfcal{M}\tilde{u}_t(1,t)- c_0 \|\tilde{u}_t(0,t) \|_2^2  + \rho_1 \int_0^1  \tilde{u}_t(1,t)^T \tilde{u}_t(\zeta,t)d\zeta - \rho_1 k_1 \tilde{u}(1,t)^T \mathbfcal{M} \tilde{u}_{}(1,t)\nonumber \\ &-\rho_1 c_0 \tilde{u}(1,t)^T \tilde{u}_{t}(0,t)-\frac{1}{2} \rho_2 \ \|\tilde{u}_\varsigma(\cdot,t)\|_{{L_2^n}}^2  -\frac{1}{2} \rho_2 \ \|\tilde{u}_t(\cdot,t)\|_{{L_2^n}}^2 + \frac{1}{2}\rho_2 (1+c_0^2) \| \tilde{u}_t(0,t)\|_2^2.    \label{VdotCA}
\end{flalign}
Let us estimate the sign-indefinite terms  in the right hand side of \eqref{VdotCA}. By means of \eqref{app1:HolderInequality} one derives the next two estimations
 \begin{flalign}
\left|\rho_1 \int_0^1  \tilde{u}_t(1,t)^T \tilde{u}_t(\zeta,t)d\zeta \right|  &\leq \frac{1}{2}\rho_1 \| \tilde{u}_t(1,t)\|_2^2 + \frac{1}{2}\rho_1\|\tilde{u}_t(\cdot,t)\|_{{L_2^n}}^2    \label{VdotCbnd1AOYT} \\ \left|\rho_1 c_0 \tilde{u}(1,t)^T \tilde{u}_{t}(0,t) \right|&\leq \frac{1}{2}\rho_1 c_0 \| \tilde{u}(1,t)\|_2^2 + \frac{1}{2}\rho_1c_0\|\tilde{u}_t(0,t)\|^2.     \label{VdotCbnd1B}
\end{flalign}
Considering \eqref{VdotCbnd1AOYT} and \eqref{VdotCbnd1B} into \eqref{VdotCA},  estimating the positive-definite quadratic forms $\tilde{u}_t^T(1,t)\mathbfcal{M}\tilde{u}_t(1,t)$ and $\tilde{u}^T(1,t)\mathbfcal{M}\tilde{u}(1,t)$
in the right-hand side of \eqref{VdotCA}, and rearranging, one obtains
{\begin{flalign}
\dot V(t) \leq & - \frac{1}{2}\rho_2 \|\tilde{u}_\varsigma(\cdot,t)\|_{{L_2^n}}^2 - \frac{1}{2} \left( \rho_2-\rho_1\right) \|\tilde{u}_t(\cdot,t)\|_{{L_2^n}}^2 - \rho_1 \left[k_1 \lambda_m(\mathbfcal{M})-\frac{1}{2}c_0\right] \| \tilde{u}(1,t)\|_2^2  \nonumber \\ &- \left[k_2 \lambda_m(\mathbfcal{M})-\frac{1}{2}\rho_1\right] \| \tilde{u}_t(1,t)\|_2^2  - \left[c_0-\frac{1}{2}\rho_1c_0-\frac{1}{2}\rho_2 (1+c_0^2)\right] \| \tilde{u}_t(0,t)\|_2^2.  \label{Vdot3CA}
\end{flalign}}
Provided that the arbitrary coefficients $\rho_1$ and $\rho_2$ meet the inequalities \eqref{ro1teo1}-\eqref{ro2teo1}, it is concluded by \eqref{Vdot3CA} and \eqref{cond14CA} that
\begin{equation}\label{vdot}
    \dot V(t) \leq - \mu V_0(t) \leq  - \frac{\mu}{\tau_2} V(t),
\end{equation}
where $\mu$ is defined in \eqref{mudeff},
which implies that $V(t)$ escapes exponentially to zero {as fast as $V(t) \leq V(0)e^{-\frac{\mu}{\tau_2} t}$. As in the proof of \cite[Theorem 1]{orlov2020}, it follows that an arbitrary error solution in question remains uniformly bounded in $H^{1,n}\times L_2^n$, and hence they are globally extendible to the right for all $t>0$.

Furthermore, by \eqref{cond14CA} one derives the estimation
\begin{equation}
 V_0(t) \leq  \frac{1}{\tau_1} V(t) \leq  \rho_0 e^{-\alpha t}, \;\;\;\rho_0= \frac{1}{\tau_1} V(0), \;\;\;\; \alpha=\frac{\mu}{\tau_2}
\end{equation}
of $V_0(t)$.}  From the definition \eqref{cond15CA} of $V_0(t)$ one concludes that the squared norms $\|\tilde{u}(1,t) \|^2_2$ and $\| \tilde{u}_\varsigma(\cdot,t)\|^2_{{L_n}}$ are both upper-estimated by $V_0(t)$.  Inequality \eqref{app1:pisanoResult}, specialized with $b(\cdot)=\tilde{u}(\cdot,t)$ and $i=1$, reads as
{\begin{flalign}
 \| \tilde{u}(\cdot,t)\|_{{L_2^n}}^2 \leq 2 \left( \|  \tilde{u}(1,t) \|_2^2 +  \| \tilde{u}_\varsigma(\cdot,t)\|_{{L_2^n}}^2 \right).  \label{Vdot3CAHH}
\end{flalign}}

{Definition \eqref{cond15CA} also implies that $$\|  \tilde{u}(1,t) \|_2^2 +  \| \tilde{u}_\varsigma(\cdot,t)\|_{{L_2^n}}^2  \leq V_0(t) \leq \rho_0 e^{-\alpha t}$$. Substituting this last estimation into \eqref{Vdot3CAHH} one obtains that
{\begin{flalign}
 \| \tilde{u}(\cdot,t)\|_{{L_2^n}}^2 \leq 2  \rho_0 e^{-\alpha t}.\label{ehiy}
\end{flalign}}
Agmon's inequality yields the next uniform estimation
\begin{equation}\label{Agm}
    \max_{x\in[0,1]}|\tilde u_i(x,t)|^2\leq \tilde u^2_i(1,t) +  \| \tilde{u}_i(\cdot,t)\|_{{L_2}} \| \tilde{u}_{i,\varsigma}(\cdot,t)\|_{{L_2}}
\end{equation}
for $|\tilde u_i(x,t)|$. The terms appearing in the right-hand side of \eqref{Agm} are estimated as
\begin{eqnarray}
  \tilde u^2_i(1,t) &\leq& \|\tilde{u}(1,t) \|^2_2\leq  \rho_0 e^{-\alpha t},  \label{b1}\\
  \| \tilde{u}_i(\cdot,t)\|_{{L_2}} &\leq&   \| \tilde{u}(\cdot,t)\|_{{L_2^n}}  \leq \sqrt{2 \rho_0 }e^{-\frac{\alpha}{2} t}, \label{b2}\\
\| \tilde{u}_{i,\varsigma}(\cdot,t)\|_{{L_2}}&\leq&  \| \tilde{u}_\varsigma(\cdot,t)\|_{{L_2^n}}  \leq \sqrt{\rho_0}e^{-\frac{\alpha}{2}t}.\label{b3}
\end{eqnarray}
Substituting \eqref{b1}-\eqref{b3} into \eqref{Agm} yields
\begin{equation}\label{Agm2}
    \max_{x\in[0,1]}|\tilde u_i(x,t)|^2\leq  \left(1+\sqrt{2}\right)\rho_0 e^{-\alpha t}
\end{equation}
which, due to definition \eqref{sect2:u_tilde} of $\tilde u_i(x,t)$, results in the point-wise consensus relation \eqref{sect1:consensus_tracking} with the parameters $\delta= \left(1+\sqrt{2}\right)\rho_0=  \frac{\left(1+\sqrt{2}\right)}{\tau_1} V(0) $ and $\alpha=\frac{\mu}{\tau_2}$.}
Since the Lyapunov functional $V(t)$ is radially unbounded, the demonstrated exponential
stability holds globally for the closed-loop system in question. Theorem \ref{lemma1:Vconv} is proven.
\end{proof}

\section{Disturbance propagation ISS analysis}\label{sectDIST}

In the sequel, the perturbed version
\begin{flalign}
u_{tt}(\varsigma,t)&=u_{\varsigma\varsigma}(\varsigma,t) + f(\varsigma,t),
\label{sect1:barDynamicPERT}\\
u(\varsigma,0)&= u^0(\varsigma), \;\;\; u_t(\varsigma,0)= u_t^0(\varsigma),
\label{sect1:barICsPERT}\\
u_\varsigma(0,t)&=c_0 u_t(0,t)+\psi_0(t), \label{sect1:freeendPERT}\\
u_\varsigma(1,t)&=q(t)+\psi_1(t)\label{sect1:actuationPERT}
\end{flalign}
of the followers' dynamics \eqref{sect1:barDynamic}-\eqref{sect1:actuation} is considered, where the PDE  \eqref{sect1:barDynamicPERT} is corrupted by an in-domain disturbance $f(\varsigma,t)$ {of class $L^{loc}_\infty(L_2)$, and the BCs \eqref{sect1:freeendPERT}-\eqref{sect1:actuationPERT}  are corrupted by boundary perturbation terms $\psi_0(t)$ and $\psi_1(t)$ of class $C^2$.
Since in the perturbed case  the compatibility condition \eqref{sect1:compa} would necessarily be modified to involve the boundary perturbations, and therefore it would  be rather restrictive, Assumption \ref{assumption:1_distrurbance} is no longer in force. Instead, the meaning of the perturbed boundary-value problem \eqref{sect1:barDynamicPERT} --\eqref{sect1:actuationPERT} is subsequently viewed in the weak sense as it is done in \cite{pisano2017} for a diffusion PDE.}


{The same local interaction control rule \eqref{sect1:consensusProthhocol0} proves to be  capable of ensuring the exponential ISS inequality, relating to suitable norms of the tracking error  \eqref{sect2:u_tilde} and admissible perturbations. In the sequel, let the arbitrary parameters $\rho_1,\rho_2$ in \eqref{VCA}-\eqref{G2CA} be such that}
\begin{flalign}
0 &< \rho_1 < \min \left( k_1\lambda_m(\mathbfcal{M}), 1-\rho_2, { \rho_2-\xi_1},  {2 k_2 \lambda_m(\mathbfcal{M})-1}, \frac{1}{c_0} \left[ 2 \left(c_0-\frac{1}{2}\xi_2\right)-\rho_2(1+c_0+c_0^2)\right]\right), \label{cond1lCAPRTB} \\
{ \xi_1} &< \rho_2 <  \min \left( 1, { \frac{2 \left(c_0-\frac{1}{2}\xi_2\right)}{1+c_0+c_0^2}}\right). \label{cond1lCA2PRTBB}
\end{flalign}
{for some  $\xi_1>0$ and  $0 < \xi_2 < \frac{1}{2c_0}$ where  the positive constant $c_0$ is the same as in the boundary conditions \eqref{sect1:freeendPERT}, \eqref{sect1:leaderBC0}.
Letting
\begin{flalign}
k_1 > \frac{c_0{+3}}{2 \lambda_m(\mathbfcal{M})} , \;\;\; k_2 > {\frac{1}{2\lambda_m(\mathbfcal{M})}} \label{cond1lCAwow}
\end{flalign}
and setting
\begin{equation}\label{vdotPERTpar}
q_0=\frac{1}{2}\left[\frac{1}{\xi_2} + \rho_1 + \rho_2(c_0+1) \right], \; q_f=\left(\frac{1}{2\xi_1} +\frac{1}{2}\rho_1 + \rho_2 \right),
\end{equation}
\begin{small}\begin{equation}\label{mu}
\mu_2 = \min \left\{ \frac{1}{4}\rho_2 ,\frac{1}{2}\left(\rho_2-\rho_1-\xi_1\right),\rho_1 \left[k_1 \lambda_m(\mathbfcal{M})-\frac{1}{2}c_0 { -\frac{3}{2} } \right]\right\}.
\end{equation}
\end{small}}
{the next result is in order.}

\begin{mythm}\label{lemma1:VconvPERT}
{Consider the  leader and perturbed follower PDEs \eqref{sect1:leader}-\eqref{sect1:leaderBC1} and \eqref{sect1:barDynamicPERT}-\eqref{sect1:actuationPERT}, initialized in $\mathrm{H}^{1,n+1}\times L_2^{n+1}$ and governed by  the local interaction rule \eqref{sect1:consensusProthhocol0}, \eqref{cond1lCAwow}. Let Assumption \ref{assumption:graph} be in force and let $f(\varsigma,t)$ be of class $L^{loc}_\infty(L_2)$ whereas $\psi_0(t)$ and $\psi_1(t)$ be of class $C^2$.
Then, the over-all closed-loop system globally possesses a unique weak solution and the following exponential ISS relation
\begin{flalign}
     V_0(t) \leq &   V_0(0) e^{-\frac{\mu_2}{\tau_2}t}  + \frac{\tau_2 q_0}{\mu_2 \tau_1}  E_s\left(\|\psi_0(t)\|_2^2\right) + \frac{\tau_2}{\mu_2\tau_1} E_s\left(\|\psi_1(t)\|_2^2\right)+  \frac{\tau_2 q_f}{\mu_2\tau_1} E_s\left(\|f(\cdot,t)\|_{L_2^n}^2\right)\label{ISS23}
\end{flalign}
holds true with  $V_0(0)$, derived from \eqref{cond15CA}, with  $\tau_1$ and $\tau_2$, given by  \eqref{tau1}, \eqref{tau}, with constants $\rho_1$ and $\rho_2$,  fulfilling \eqref{cond1lCAPRTB},\eqref{cond1lCA2PRTBB}, and with $q_0$, $q_f$, $\mu_2$, specified by \eqref{vdotPERTpar}, \eqref{mu}.}

\end{mythm}

\vspace{.15cm}

\begin{proof} {By following the line of reasoning used in the proof of \cite[Theorem 1]{pisano2017}, the existence of a unique weak solution of the closed-loop system \eqref{sect1:leader}-\eqref{sect1:leaderBC1}, \eqref{sect1:consensusProthhocol0}, \eqref{sect1:barDynamicPERT}-\eqref{sect1:actuationPERT} is established.
The same Lyapunov functional \eqref{VCA}-\eqref{G2CA}, adopted to analyze the unperturbed dynamics, is now applied to the ISS analysis.} Provided that restrictions  \eqref{cond1lCAPRTB}-\eqref{cond1lCA2PRTBB} hold, the previously derived relations 
\begin{flalign}
\tau_1 V_0(t) \leq V(t) \leq \tau_2 V_0(t),\label{cond14CATEOPERT}
\end{flalign}
are still in force, where $V_0(t)$ is defined in \eqref{cond15CA}, and the positive constants $\tau_1$ and $\tau_2$ are defined in \eqref{tau1}-\eqref{tau}. {By \eqref{sect1:barDynamicPERT}-\eqref{sect1:actuationPERT} and \eqref{sect1:leader}-\eqref{sect1:leaderBC1}, coupled to the local interaction rule \eqref{sect1:consensusProthhocol0}, the boundary value problem, governing the deviation variable \eqref{sect2:u_tilde}, now reads as
\begin{flalign}
&\tilde{u}_{tt}(\varsigma,t)=\tilde{u}_{\varsigma\varsigma}(\varsigma,t)+f(\varsigma,t), \label{sect2:errDynPERT}\\
&\tilde{u}(\varsigma,0)={u}(\varsigma,0)-1_n u_0(\varsigma,0),\\
&\tilde{u}_t(\varsigma,0)={u}_{t}(\varsigma,0)-1_n u_{0,t}(\varsigma,0),\label{sect2:errICPERT}\\
&\tilde{u}_{\varsigma}(0,t)=c_0 \tilde{u}_{t}(0,t) + \psi_0(t),\label{sect2:errBC0PERT} \\
&\tilde{u}_{\varsigma}(1,t)=-k_1 \mathbfcal{M}\tilde{u}(1,t)-k_2 \mathbfcal{M}\tilde{u}_t(1,t)+\psi_1(t).
\label{sect2:errBCPERT}
\end{flalign}}
Let us now compute the time derivative of $V(t)$ along the weak solutions of \eqref{sect2:errDynPERT}-\eqref{sect2:errBCPERT}.
Differentiating \eqref{ECA} and substituting \eqref{sect2:errDynPERT}  in the resulting expression yields
\begin{flalign}
\dot E(t)=& \int_0^1 \tilde{u}_\varsigma(\zeta,t)^T \tilde{u}_{\varsigma t}(\zeta,t)d\zeta + \int_0^1 \tilde{u}_t(\zeta,t)^T \tilde{u}_{t t}(\zeta,t)d\zeta  + k_1 \tilde{u}^T(1,t)\mathbfcal{M}\tilde{u}_t(1,t)   \nonumber \\
=&\int_0^1 \tilde{u}_\varsigma(\zeta,t)^T \tilde{u}_{\varsigma t}(\zeta,t)d\zeta + \int_0^1 \tilde{u}_t(\zeta,t)^T \tilde{u}_{\varsigma \varsigma}(\zeta,t)d\zeta  + \int_0^1 \tilde{u}_t(\zeta,t)^T f(\zeta,t)d\zeta + k_1 \tilde{u}^T(1,t)\mathbfcal{M}\tilde{u}_t(1,t). \label{EdotCALONGpeti1}
\end{flalign}
Performing integration by parts, considering \eqref{sect2:errBC0PERT}-\eqref{sect2:errBCPERT} into the resulting relation, and finally rearranging, yield the chain of equalities

\begin{flalign}
\dot E(t)=&\int_0^1 \tilde{u}_\varsigma(\zeta,t)^T \tilde{u}_{\varsigma t}(\zeta,t)d\zeta + \left. \tilde{u}_t(x,t)^T \tilde{u}_{\varsigma}(x,t) \right|_{x=0}^{x=1} \nonumber \\
& -\int_0^1 \tilde{u}_\varsigma(\zeta,t)^T \tilde{u}_{\varsigma t}(\zeta,t)d\zeta +  k_1  \tilde{u}_t^T(1,t)\mathbfcal{M}\tilde{u}(1,t)  + \int_0^1 \tilde{u}_t(\zeta,t)^T f(\zeta,t)d\zeta  \nonumber \\
= & \; \tilde{u}_t(1,t)^T \tilde{u}_{\varsigma}(1,t)-\tilde{u}_t(0,t)^T \tilde{u}_{\varsigma}(0,t)  +  k_1  \tilde{u}_t^T(1,t)\mathbfcal{M}\tilde{u}(1,t) + \int_0^1 \tilde{u}_t(\zeta,t)^T f(\zeta,t)d\zeta  \nonumber \\
=&- k_2\tilde{u}_t^T(1,t)\mathbfcal{M}\tilde{u}_t(1,t)- c_0 \|\tilde{u}_t(0,t) \|_2^2 + \tilde{u}_t^T(1,t)\psi_1(t) - \tilde{u}_t^T(0,t)\psi_0(t) + \int_0^1 \tilde u_t(\zeta,t)f(\zeta,t)d\zeta,  \label{EdotCALONGpeti2}
\end{flalign}
Differentiating \eqref{G1CA} and and substituting \eqref{sect2:errDynPERT} in the resulting expression yields
\begin{flalign}
\dot G_1(t) & =\rho_1 k_2 \tilde{u}_t(1,t)^T\mathbfcal{M}\tilde{u}(1,t)+ \rho_1 \int_0^1  \tilde{u}_t(1,t)^T \tilde{u}_t(\zeta,t)d\zeta +  \rho_1 \int_0^1  \tilde{u}(1,t)^T \tilde{u}_{tt}(\zeta,t)d\zeta\nonumber \\ & =\rho_1 k_2 \tilde{u}_t(1,t)^T\mathbfcal{M}\tilde{u}(1,t) + \rho_1 \int_0^1  \tilde{u}_t(1,t)^T \tilde{u}_t(\zeta,t)d\zeta +  \rho_1 \int_0^1  \tilde{u}(1,t)^T \tilde{u}_{\varsigma \varsigma}(\zeta,t)d\zeta \nonumber \\
 & +  \rho_1 \int_0^1  \tilde{u}(1,t)^T f(\zeta,t)d\zeta. \label{G1dotinitCApt00}
\end{flalign}
By direct integration and considering \eqref{sect2:errBC0PERT}-\eqref{sect2:errBCPERT} it yields
\begin{flalign}
 \rho_1 \int_0^1  \tilde{u}(1,t)^T \tilde{u}_{\varsigma \varsigma}(\zeta,t)d\zeta =& \rho_1  \tilde{u}(1,t)^T \left. \tilde{u}_{\varsigma}(x,t) \right|_{x=0}^{x=1} \nonumber \\ =& \rho_1  \tilde{u}(1,t)^T \tilde{u}_{\varsigma}(1,t) -  \rho_1  \tilde{u}(1,t)^T \tilde{u}_{\varsigma}(0,t) \nonumber \\
 =& - \rho_1 k_1 \tilde{u}(1,t)^T \mathbfcal{M} \tilde{u}_{}(1,t)- \rho_1 k_2 \tilde{u}(1,t)^T \mathbfcal{M} \tilde{u}_{t}(1,t)+\rho_1 \tilde u^T(1,t)\psi_1(t)\nonumber \\ &-\rho_1 c_0 \tilde{u}(1,t)^T \tilde{u}_{t}(0,t) -\rho_1 \tilde u ^T(1,t)\psi_0(t).
 \label{G1dotinitCApt2}
\end{flalign}
Substituting \eqref{G1dotinitCApt2} into \eqref{G1dotinitCApt00} one obtains
 \begin{flalign}
\dot G_1(t) = & \rho_1 \int_0^1  \tilde{u}_t(1,t)^T \tilde{u}_t(\zeta,t)d\zeta - \rho_1 k_1 \tilde{u}(1,t)^T \mathbfcal{M} \tilde{u}_{}(1,t)-\rho_1 c_0 \tilde{u}(1,t)^T \tilde{u}_{t}(0,t) \nonumber \\ & + \rho_1 \int_0^1  \tilde{u}(1,t)^T f(\zeta,t)d\zeta  +\rho_1 \tilde u^T(1,t)\psi_1(t)-\rho_1 \tilde u ^T(1,t)\psi_0(t).  \label{G1dotinDFitC2222A}
\end{flalign}

Differentiating \eqref{G2CA} and substituting \eqref{sect2:errDynPERT} yields
\begin{flalign}
\dot G_2(t) =& \rho_2  \int_0^1 (\zeta-1) \tilde{u}_{tt}(\zeta,t)^T \tilde{u}_\varsigma(\zeta,t)d\zeta  +  \rho_2  \int_0^1 (\zeta-1) \tilde{u}_{t}(\zeta,t)^T \tilde{u}_{\varsigma t}(\zeta,t)d\zeta = \nonumber \\
=&\rho_2  \int_0^1 (\zeta-1) \tilde{u}_{\varsigma \varsigma}(\zeta,t)^T \tilde{u}_\varsigma(\zeta,t)d\zeta   + \rho_2  \int_0^1 (\zeta-1) f(\zeta,t)^T \tilde{u}_\varsigma(\zeta,t)d\zeta  +  \rho_2  \int_0^1 (\zeta-1) \tilde{u}_{t}(\zeta,t)^T \tilde{u}_{\varsigma t}(\zeta,t)d\zeta   \label{G2dotinitCApreliPERT2}
\end{flalign}
Integrating by parts one derives
\begin{flalign}
\dot G_2(t) =&\left. \frac{\rho_2}{2}(x-1)\tilde{u}_\varsigma(x,t)^T\tilde{u}_\varsigma(x,t) \right|_{x=0}^{x=1}- \frac{1}{2} \rho_2 \|\tilde{u}_\varsigma(\cdot,t)\|_{{L_2^n}}^2   + \rho_2  \int_0^1 (\zeta-1) f(\zeta,t)^T \tilde{u}_\varsigma(\zeta,t)d\zeta \nonumber \\
 & +  \left. \frac{\rho_2}{2}(x-1)\tilde{u}_t(x,t)^T\tilde{u}_t(x,t) \right|_{x=0}^{x=1}- \frac{1}{2} \rho_2 \|\tilde{u}_t(\cdot,t)\|_{{L_2^n}}^2 \nonumber \\
  =& \; \frac{\rho_2}{2}\tilde{u}_\varsigma(0,t)^T\tilde{u}_\varsigma(0,t) - \frac{1}{2} \rho_2 \|\tilde{u}_\varsigma(\cdot,t)\|_{{L_2^n}}^2   + \rho_2  \int_0^1 (\zeta-1) f(\zeta,t)^T \tilde{u}_\varsigma(\zeta,t)d\zeta  +   \frac{\rho_2}{2} \tilde{u}_t(0,t)^T\tilde{u}_t(0 ,t)  - \frac{1}{2} \rho_2 \|\tilde{u}_t(\cdot,t)\|_{{L_2^n}}^2. \label{G2dotinitCAPERTo1pre}
\end{flalign}

Substituting  \eqref{sect2:errBC0PERT} into \eqref{G2dotinitCAPERTo1pre} and rearranging one obtains
\begin{flalign}
\dot G_2(t) =& - \frac{1}{2} \rho_2 \|\tilde{u}_\varsigma(\cdot,t)\|_{{L_2^n}}^2- \frac{1}{2} \rho_2 \|\tilde{u}_t(\cdot,t)\|_{{L_2^n}}^2   +  \rho_2  \int_0^1 (\zeta-1) f(\zeta,t)^T \tilde{u}_\varsigma(\zeta,t)d\zeta \nonumber \\
 & +\frac{1}{2} \rho_2 (1+c_0^2) \|\tilde{u}_t(0,t) \|_2^2 + \rho_2 c_0 \psi_0^T(t)\tilde{u}_t(0,t)  +\frac{1}{2} \rho_2 \|\psi_0(t)\|_2^2.  \label{G2dotinitCAPERTo1}
\end{flalign}

By using relation \eqref{app1:HolderInequality}, let us estimate all perturbation-dependent and sign-indefinite terms in the right-hand sides of \eqref{EdotCALONGpeti2}, \eqref{G1dotinDFitC2222A},  and \eqref{G2dotinitCAPERTo1}.
\begin{flalign}
 |\tilde{u}_t^T(1,t)\psi_1(t)| \leq  \frac{1}{2}\|\tilde{u}_t(1,t)\|_2^2+\frac{1}{2}\|\psi_1(t)\|_2^2  \label{pertestNEW} \\
 |\tilde{u}_t^T(0,t)\psi_0(t)| \leq  \frac{\xi_2}{2}\|\tilde{u}_t(0,t)\|_2^2+\frac{1}{2\xi_2}\|\psi_0(t)\|_2^2, \;\; \xi_2>0  \label{pertestNEW2}
\end{flalign}
\begin{flalign}
 \left|\int_0^1 \tilde u_t(\zeta,t)f(\zeta,t)d\zeta\right| \leq  \frac{\xi_1}{2}\|\tilde{u}_t(\cdot,t)\|_{L_2^n}^2+\frac{1}{2\xi_1}\|f(\cdot,t)\|_{L_2^n}^2, \; \xi_1>0  \label{pertestNEW3}
\end{flalign}
\begin{flalign}
 \left|\rho_1 \int_0^1  \tilde{u}_t(1,t)^T \tilde{u}_t(\zeta,t)d\zeta\right| \leq  \frac{\rho_1}{2}\|\tilde{u}_t(1,t)\|_{2}^2+\frac{\rho_1}{2}\|\tilde{u}_t(\cdot,t)\|_{L_2^n}^2  \label{pertestNEW4}
\end{flalign}
\begin{flalign}
 \left| \rho_1 c_0 \tilde{u}(1,t)^T \tilde{u}_{t}(0,t)  \right| \leq  \frac{\rho_1c_0}{2}\|\tilde{u}(1,t)\|_{2}^2+\frac{\rho_1c_0}{2}\|\tilde{u}_t(0,t)\|_{2}^2  \label{pertestNEW5}
\end{flalign}
\begin{flalign}
 \left| \rho_1 \int_0^1  \tilde{u}(1,t)^T f(\zeta,t)d\zeta   \right| \leq  \frac{\rho_1}{2}\|\tilde{u}(1,t)\|_{2}^2+\frac{\rho_1}{2}\|f(\cdot,t)\|_{L_2^n}^2  \label{pertestNEW6}
\end{flalign}
\begin{flalign}
 |\rho_1\tilde{u}^T(1,t)\psi_1(t)| \leq  \frac{\rho_1}{2}\|\tilde{u}(1,t)\|_2^2+\frac{\rho_1}{2}\|\psi_1(t)\|_2^2  \label{pertestNEW7} \\
 |\rho_1\tilde{u}^T(1,t)\psi_0(t)| \leq  \frac{\rho_1}{2}\|\tilde{u}(1,t)\|_2^2+\frac{\rho_1}{2}\|\psi_0(t)\|_2^2  \label{pertestNEW8}
\end{flalign}
 \begin{flalign}
 &\left| \rho_2  \int_0^1 (\zeta-1) f(\zeta,t)^T \tilde{u}_\varsigma(\zeta,t)d\zeta  \right|  \leq \left| \rho_2  \int_0^1  f(\zeta,t)^T \tilde{u}_\varsigma(\zeta,t)d\zeta  \right|  \leq  \rho_2\|f(\cdot,t)\|_{L_2^n}^2+\frac{\rho_2}{4}\|\tilde{u}_\varsigma(\cdot,t)\|_{L_2^n}^2  \label{pertestNEW9}
\end{flalign}
\begin{flalign}
 |\rho_2 c_0 \psi_0^T(t)\tilde{u}_t(0,t)| \leq  \frac{\rho_2c_0}{2}\|\tilde{u}_t(0,t)\|_2^2+\frac{\rho_2c_0}{2}\|\psi_0(t)\|_2^2  \label{pertestNEW10}
\end{flalign}

 Substituting \eqref{pertestNEW}-\eqref{pertestNEW3} into \eqref{EdotCALONGpeti2}, estimating the quadratic form, and rearranging yields
\begin{flalign}
\dot E(t)\leq&- k_2  \lambda_m(\mathbfcal{M}) \| \tilde{u}_t(1,t) \|_2^2 - c_0 \|\tilde{u}_t(0,t) \|_2^2  + \frac{1}{2}\|\tilde{u}_t(1,t)\|_2^2+\frac{1}{2}\|\psi_1(t)\|_2^2 \nonumber \\
& + \frac{\xi_2}{2}\|\tilde{u}_t(0,t)\|_2^2+\frac{1}{2\xi_2}\|\psi_0(t)\|_2^2   + \frac{\xi_1}{2}\|\tilde{u}_t(\cdot,t)\|_{L_2^n}^2+\frac{1}{2\xi_1}\|f(\cdot,t)\|_{L_2^n}^2. \label{EdotCALONGpeti2SUBBND}
\end{flalign}

Substituting \eqref{pertestNEW4}-\eqref{pertestNEW8} into \eqref{G1dotinDFitC2222A}, estimating the quadratic form, and rearranging yields
  \begin{flalign}
\dot G_1(t) \leq &   \frac{\rho_1}{2}\|\tilde{u}_t(1,t)\|_{2}^2+\frac{\rho_1}{2}\|\tilde{u}_t(\cdot,t)\|_{L_2^n}^2 - \rho_1 k_1 \lambda_m(\mathbfcal{M}) \|\tilde{u}(1,t)\|_2^2   +\frac{\rho_1c_0}{2}\|\tilde{u}(1,t)\|_{2}^2+\frac{\rho_1c_0}{2}\|\tilde{u}_t(0,t)\|_{2}^2  \nonumber \\
&  +\frac{\rho_1}{2}\|\tilde{u}(1,t)\|_{2}^2+\frac{\rho_1}{2}\|f(\cdot,t)\|_{L_2^n}^2  + \frac{\rho_1}{2}\|\tilde{u}(1,t)\|_2^2+\frac{\rho_1}{2}\|\psi_1(t)\|_2^2   + \frac{\rho_1}{2}\|\tilde{u}(1,t)\|_2^2+\frac{\rho_1}{2}\|\psi_0(t)\|_2^2. \label{G1dot2ABNDD}
\end{flalign}

Substituting \eqref{pertestNEW9}-\eqref{pertestNEW10} into \eqref{G2dotinitCAPERTo1}, and rearranging yields
\begin{flalign}
\dot G_2(t) \leq & - \frac{1}{2} \rho_2 \|\tilde{u}_\varsigma(\cdot,t)\|_{{L_2^n}}^2- \frac{1}{2} \rho_2 \|\tilde{u}_t(\cdot,t)\|_{{L_2^n}}^2   + \rho_2\|f(\cdot,t)\|_{L_2^n}^2+\frac{\rho_2}{4}\|\tilde{u}_\varsigma(\cdot,t)\|_{L_2^n}^2   \nonumber \\
 & +\frac{1}{2} \rho_2 (1+c_0^2) \|\tilde{u}_t(0,t) \|_2^2 +\frac{1}{2} \rho_2 \|\psi_0(t)\|_2^2  + \frac{\rho_2c_0}{2}\|\tilde{u}_t(0,t)\|_2^2+\frac{\rho_2c_0}{2}\|\psi_0(t)\|_2^2.   \label{G2tinitCRToBNDD1}
\end{flalign}

Combining together \eqref{EdotCALONGpeti2SUBBND}-\eqref{G2tinitCRToBNDD1} one obtains after some straightforward manipulations the estimation
\begin{small}
\begin{flalign}
\dot V(t) = & \; \dot E(t)+\dot G_1(t)+\dot G_2(t) \nonumber \\
\leq & - \frac{1}{4}\rho_2 \|\tilde{u}_\varsigma(\cdot,t)\|_{{L_2^n}}^2 - \frac{1}{2} \left( \rho_2-\rho_1{-\xi_1}\right) \|\tilde{u}_t(\cdot,t)\|_{{L_2^n}}^2- \rho_1 \left[k_1 \lambda_m(\mathbfcal{M})-\frac{1}{2}c_0{ -\frac{3}{2} } \right] \| \tilde{u}(1,t)\|_2^2  \nonumber \\ &- \left[k_2 \lambda_m(\mathbfcal{M}){ -\frac{1}{2}} -\frac{1}{2}\rho_1\right] \| \tilde{u}_t(1,t)\|_2^2  - \left[c_0{ -\frac{1}{2}\xi_2} -\frac{1}{2}\rho_1c_0-\frac{1}{2}\rho_2 (1+c_0+c_0^2)\right] \| \tilde{u}_t(0,t)\|_2^2 \nonumber \\
& +\frac{1}{2}\left[\frac{1}{\xi_2} + \rho_1 + \rho_2(c_0+1) \right] \|\psi_0(t)\|_2^2+ \|\psi_1(t)\|_2^2  + \left(\frac{1}{2}\xi_1 +\frac{1}{2}\rho_1 + \rho_2 \right) \|f(\cdot,t)\|_{L_2^n}^2. \label{Vdot3CAPERT}
\end{flalign}
\end{small}

Provided that conditions \eqref{cond1lCAwow} hold and the arbitrary coefficients $\rho_1$ and $\rho_2$ meet the inequalities \eqref{cond1lCAPRTB}-\eqref{cond1lCA2PRTBB}, it is therefore concluded by \eqref{Vdot3CAPERT} that
\begin{equation}\label{vdotPERT1}
    \dot V(t) \leq  -  \mu_2 V_0(t)+ q_0 \|\psi_0(t)\|_2^2+\|\psi_0(t)\|_2^2+q_f\|f(\cdot,t)\|_{L_2^n}^2
\end{equation}
where $\mu_2>0$ is defined in \eqref{mu} and parameters $q_0$, $q_f$ are defined in \eqref{vdotPERTpar}. By virtue of \eqref{cond15CA} and \eqref{cond14CATEOPERT}, estimation \eqref{vdotPERT1} yields
\begin{equation}\label{vdotPERT}
    \dot V(t) \leq - \frac{\mu_2}{\tau_2} V(t)+ q_0 \|\psi_0(t)\|_2^2+\|\psi_0(t)\|_2^2+q_f\|f(\cdot,t)\|_{L_2^n}^2,
\end{equation}
where $\tau_2>0$ is given in \eqref{tau}. Applying the Comparison Lemma 3.4 from \cite{khalil2002nonlinear} to \eqref{vdotPERT} yields the ISS property
\begin{flalign}
     V(t) \leq &   V(0) e^{-\frac{\mu_2}{\tau_2}t}  + \frac{\tau_2 q_0}{\mu_2}  E_s\left(\|\psi_0(t)\|_2^2\right) + \frac{\tau_2}{\mu_2} E_s\left(\|\psi_1(t)\|_2^2\right)+  \frac{\tau_2 q_f}{\mu_2} E_s\left(\|f(\cdot,t)\|_{L_2^n}^2\right).\label{vdotPERDDDDT}
\end{flalign}
It follows that the weak solutions of the error system \eqref{sect2:errDynPERT}-\eqref{sect2:errBCPERT} do not escape to infinity in finite time. Hence, these solutions are globally extendible to the right for all $t>0$. To complete the proof it remains to note that coupling  \eqref{vdotPERDDDDT} to \eqref{cond14CATEOPERT} it directly follows \eqref{ISS23}. Theorem \ref{lemma1:VconvPERT} is proven.
\end{proof}
\vspace{.1cm}
\begin{myrem}\label{rk4iss}
By exploiting Lemma \ref{lemma1:pisanoResult} and  considering definition \eqref{cond15CA} one derives the inequality $\|\tilde{u}(\cdot,t)\|_{{L_2^n}}^2 \leq 2 V_0(t)$,
due to which the ISS inequality \eqref{ISS23} straightforwardly yields a similar estimation
directly involving the tracking error norm $\|\tilde{u}(\cdot,t)\|_{{L_2^n}}^2 $. 
\end{myrem}
\vspace{.1cm}


\section{Simulation Results}\label{sect4}

A network including one leader and three followers is considered, with the communication topology displayed in Fig.~\ref{topology} (where agent 0 represents the leader). Matrix $\mathbfcal{M}$ is
\begin{equation}
  \mathbfcal{M}=\begin{bmatrix}
                2 & - 1 & 0 \\
                -1 & 2 & -1 \\
                0 & -1 & 1 \\
              \end{bmatrix}.
\end{equation}
whose minimum and maximum eigenvalues are $\lambda_m(\mathbfcal{M})=0.19$ and $\lambda_M(\mathbfcal{M})=3.2$. The $c_0$ parameter entering the leader's and followers' boundary condition is set as $c_0=2.5$. The initial agents' transversal displacement is $u^0(\varsigma)=[5\cos(2 \pi \varsigma) , \cos(\pi \varsigma), -5 \cos(\pi \varsigma)]$, for the followers, and $u_0(x,0)=10\cos(2\pi \varsigma)$ for the leader, whereas the initial agents' velocities are $u^0_{t}(\varsigma)=[\varsigma,2\varsigma,3\varsigma]$ for the followers and $u_{0,t}(\varsigma)=0$ for the leader.
The boundary control gains were set as $k_1=30$ and $k_2=10$ in accordance with \eqref{cond1lCA} and \eqref{cond1lCAwow}.  We ran three simulation tests.  In Test 1, the case where no perturbations affect agent's dynamics ($\psi_0(t)=\psi_1(t)=0_3$ and $f(\zeta,t)=0_3$) was considered. Figure~\ref{fig:01} reports the results of Test 1 by showing the norm of the tracking error error and the spatio-temporal profile of the deviation between the state of the leader and that of follower one, which both confirm that the states of all follower agents will be synchronizing with that of the leader.
 In the Test 2 and Test 3, two distinct perturbed situations are studied. Particularly, in Test 2, we consider the perturbations $\psi_0(t)=\psi_1(t)=10\cos(10t)1_3$ and $f(\varsigma,t)=10\cos(10t) 1_3$ whereas in Test 3 they are set as $\psi_0(t)=\psi_1(t)=50\cos(10t)1_3$ and $f(\varsigma,t)=50\cos(10t)1_3$. Figure~\ref{fig:02} shows the norm of the tracking error in the Test 2 and Test 3. As expected, in the steady state the norm shown in Figure~\ref{fig:02}-right is five times higher than that of Fig.~\ref{fig:02}-left. 
\begin{figure}
  \centering
  \includegraphics[width=3in]{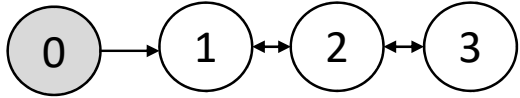}
  \caption{Considered communication topology.}\label{topology}
\end{figure}

\begin{figure}
  \centering
  \includegraphics[width=3in]{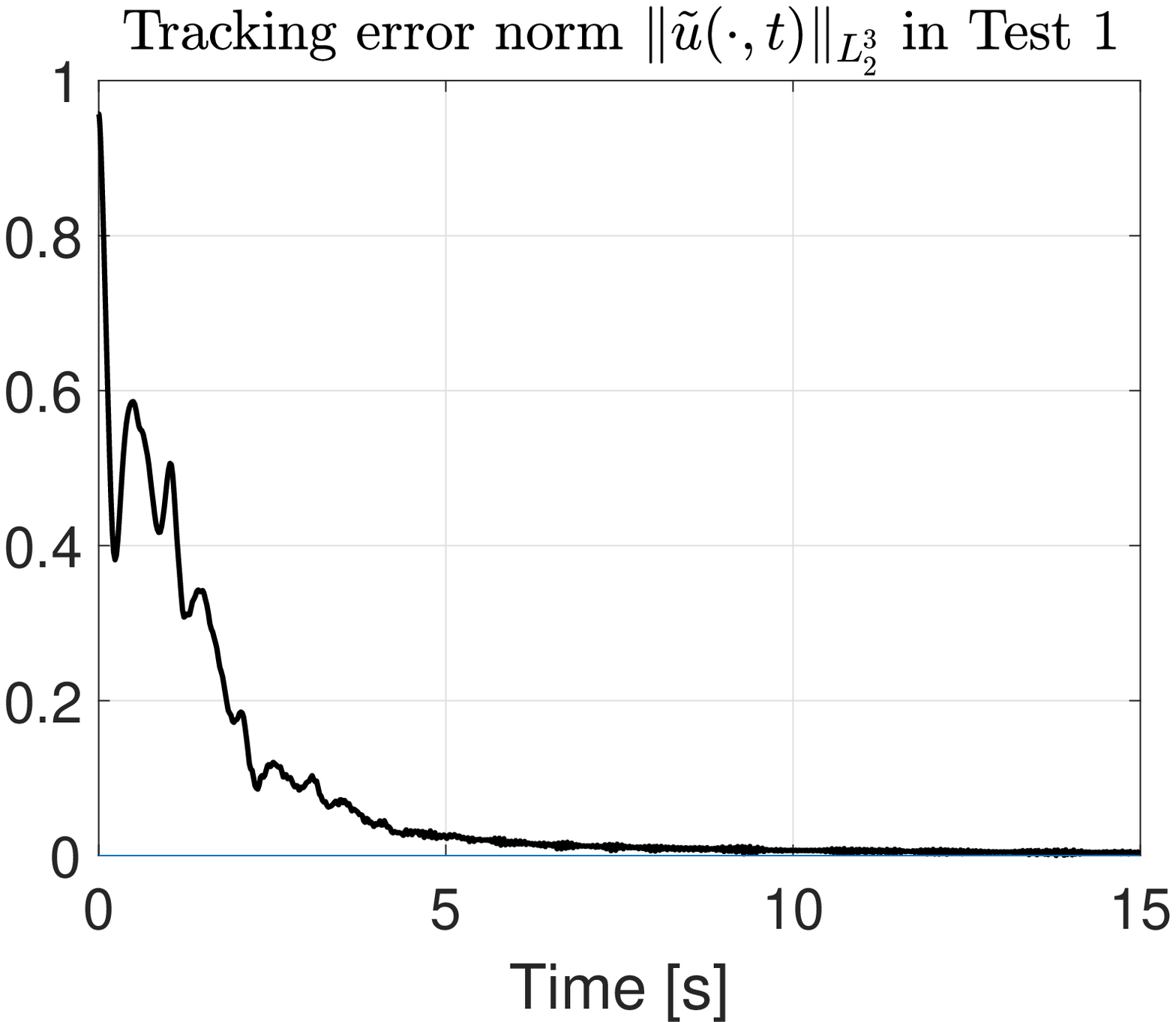}
  \includegraphics[width=3in]{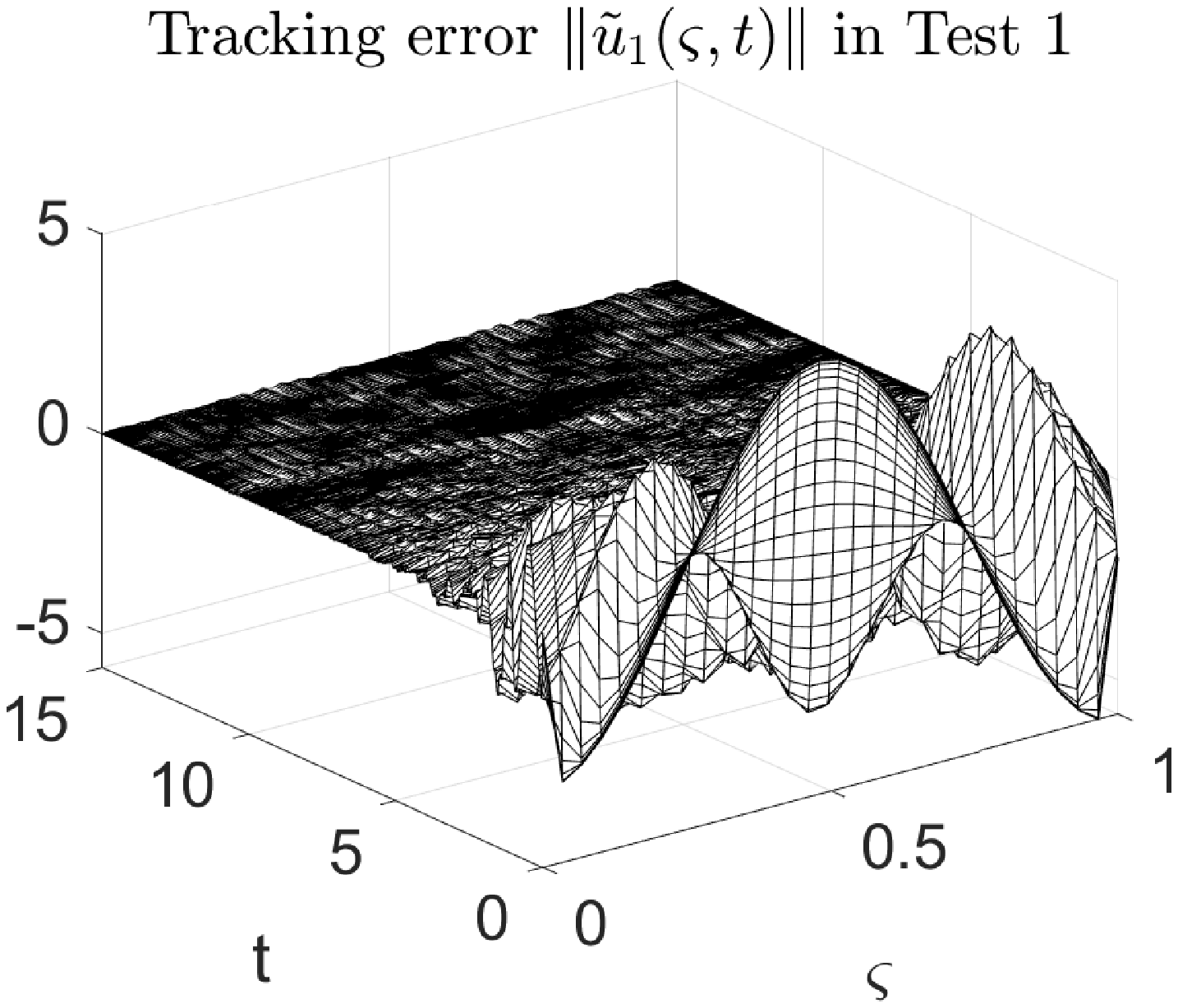}
  \caption{Results  of Test 1: (left) $L_2^3$-norm of the tracking error; (right) Spatio-temporal profile of the error variable $\tilde{u}_1(\varsigma,t)$.}\label{fig:01}
\end{figure}

\begin{figure}
  \centering
  \includegraphics[width=3in]{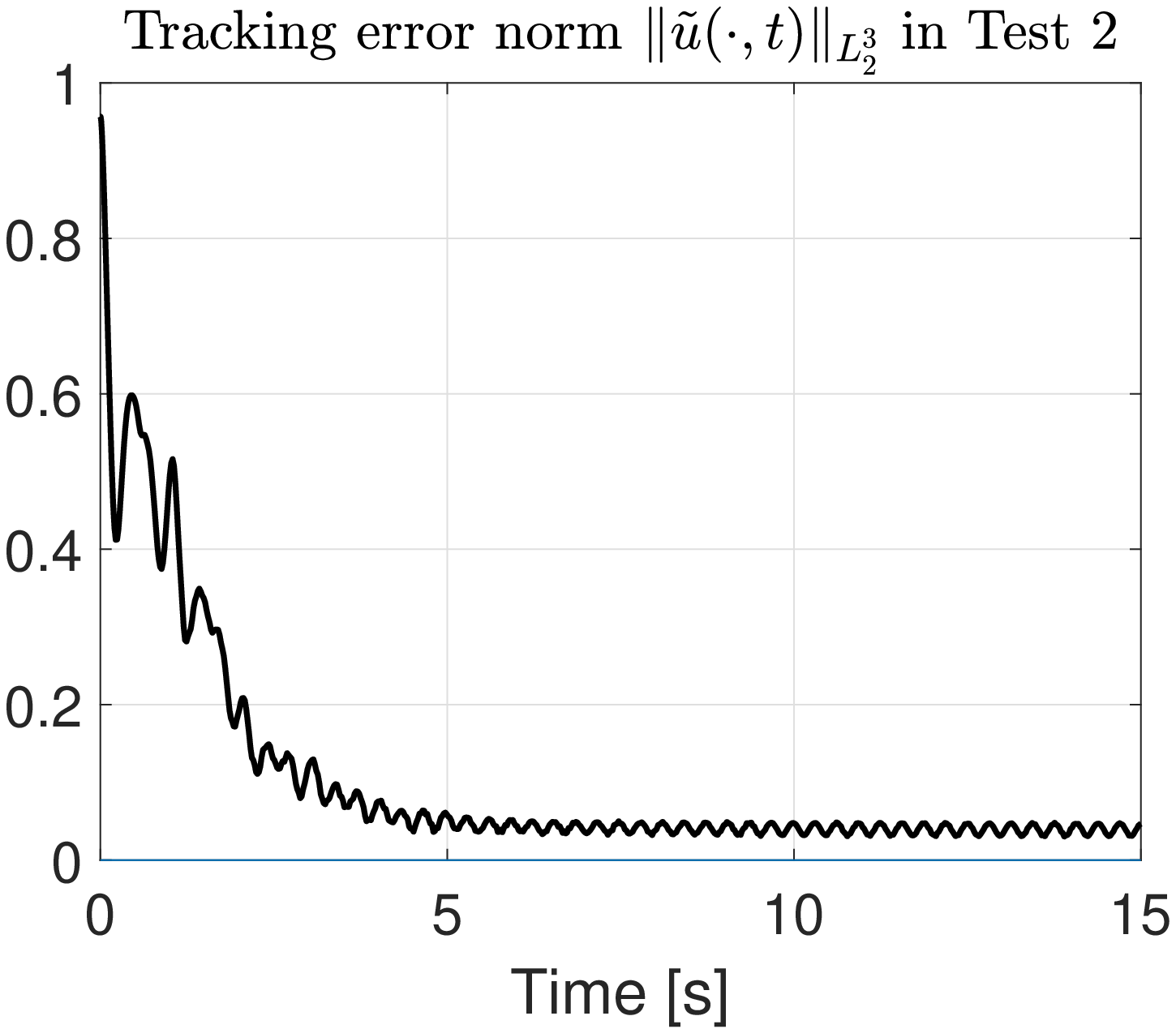}
  \includegraphics[width=3in]{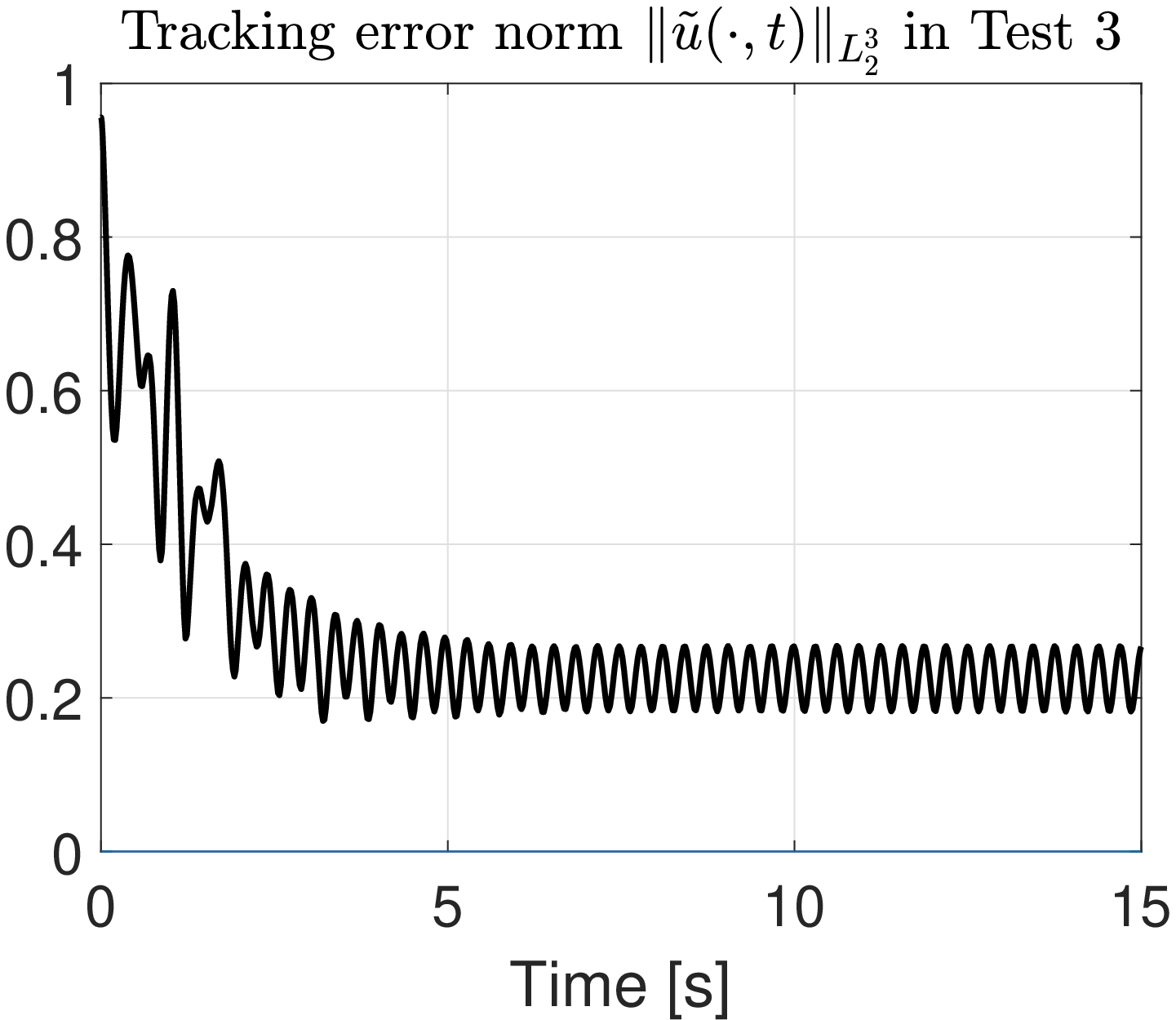}
  \caption{$L_2^3$-norm of the tracking error in Test 2 (left) and Test 3 (right).}\label{fig:02}
\end{figure}

\section{Conclusion}\label{sect5}

In this paper, the leader follower consensus problem  has been addressed with reference to agents's dynamics modeled by boundary-controlled wave equations. Among the future problems to be tackled, the (leaderless or leader-following) consensus problem for distributed parameter multi-agent systems with directed, and possibly switching, communication topology appears to be interesting and meaningful. 

%



%


\begin{thebibliography}{00}



 \bibitem{cao2012distributed} Y. Cao, and W. Ren, ``Distributed coordinated tracking with reduced interaction via a variable structure approach,'' IEEE Trans. Aut. Contr., vol. 57, pp. 33--48, January 2012.

 \bibitem{BipartiteSCL2020} Y. Chen, Z. Zuo, and Y. Wang, ``Bipartite consensus for a network of wave equations with time-varying disturbances,'' Syst. Contr. Lett., vol. 136, 104604, February 2020.

 \bibitem{curtain1995introduction} R. Curtain, and H. Zwart, An introduction to infinite-dimensional linear systems theory, Springer, 1995.


\bibitem{demetriou2013synchronization} M.A. Demetriou, ''Synchronization and consensus controllers for a class of parabolic distributed parameter systems,'' Syst. Contr. Lett., vol. 62, n. 1, pp. 70-76, 2013.

 \bibitem{demetriou2018Outf} M.A. Demetriou, ''Design of adaptive output feedback synchronizing controllers for networked PDEs with boundary and in-domain structured perturbations and disturbances,'' Syst. Contr. Lett., vol. 90, pp. 220-229, 2018.

 \bibitem{ConsensusWaveNeurocomput2018}  Q. Fu, L. Du, G. Xu, J. Wu, and P. Yu, ``Consensus control for multi-agent systems with distributed parameter models,'' Neurocomputing, vol. 308, pp. 58-64, 2018.

      \bibitem{HeIET2018} P. He, ``Consensus of uncertain parabolic PDE  via adaptive unit-vector
control scheme,'' IET Control Theory Appl., vol. 12, pp. 2488--2494, 2018.


 \bibitem{JadbaMOrseTAC03} A. Jadbabaie, J. Lin, and A. S. Morse, ''Coordination of Groups of Mobile Autonomous Agents Using Nearest Neighbor Rule,'' IEEE Trans. Aut. Contr. vol. 48, pp. 988-1001, 2003.



 \bibitem{khalil2002nonlinear} H.K. Khalil, Nonlinear Systems, 3rd ed.,Prentice Hall, 2002.

 \bibitem{li2014exact} T. Li, and B. Rao. ''Exact synchronization for a coupled system of wave equations with Dirichlet boundary controls,'' Chinese Annals of Mathematics,  vol. 34, pp. 139-160, 2013.

 \bibitem{li2018exact} T. Li, X. Lu, and B. Rao. ''Exact synchronization for a coupled system of wave equations with Neumann boundary controls,'' Chinese Annals of Mathematics,  vol. 39, pp. 233-252, 2018.

 \bibitem{LiRenTAC2014} Z. Li, G. Wen, Z. Duan, and W. Ren, ``Designing Fully Distributed Consensus Protocols for Linear Multi-Agent Systems With Directed Graphs,'' IEEE Trans. Aut. Contr., vol. 60, n. 4, pp. 1152--1157, 2014.

\bibitem{PilloConsCPDELeaderFollowHeat2016} Y. Orlov, A. Pilloni, A. Pisano, and E. Usai, ``Consensus-based leader-follower tracking for a network of perturbed diffusion PDEs via local boundary interaction,'' Proc. 2nd IFAC Workshop on Control of Systems Governed by Partial Differential Equations CPDE 2016: Bertinoro, Italy, IFAC-PapersOnLine, vol. 49, no. 8, pp. 228-233, 2016

\bibitem{orlov2020} Y. Orlov, A. Fradkov, and B. Andrievsky ``Output Feedback Energy Control of the Sine-Gordon PDE Model Using Collocated Spatially Sampled Sensing and Actuation,'' IEEE Trans. Aut. Contr., vol. 65, n. 4, 1484--1498, 2020



 \bibitem{PilloConsTACHeat} A. Pilloni, A. Pisano, Y. Orlov, and E. Usai, ``Consensus-based control for a
network of diffusion PDEs with boundary local interaction,'' IEEE Trans.
Aut. Contr., vol. 61, no. 9, pp. 2708--2713, 2016.

\bibitem{pisano2017} A. Pisano, and Y. Orlov, ``On the ISS properties of a class of parabolic DPS’ with discontinuous
control using sampled-in-space sensing and actuation,'' Automatica, vol. 81,  pp. 447--454, 2015.


\bibitem{renijrnc07natural} W. Ren and E. M. Atkins, ``Distributed multi-vehicle coordinated control via local information exchange,'' Int. J. Robust Nonlinear Contr., vol. 17, no. 10-11,  pp. 1002--1033, 2007.

\bibitem{SONGSCL2009} Q. Song, J. Cao, and W. Yu ``Second-order leader-following consensus of nonlinear multi-agent systems via pinning control,'' Syst. Contr. Lett., vol. 59, no. 9, pp. 553-562, 2010.

\bibitem{KRSTICSCL2009} A. Smyshlyaev, and M.Krstic, ``Boundary control of an anti-stable wave equation with anti-damping on the
uncontrolled boundary,'' Syst. Contr. Lett., vol. 58, pp. 617-623, 2009.

\bibitem{ZhouSCL2012} J. Zhou, X. Wu, W. Yu, M. Small, and J. Lu, ``Flocking of multi-agent dynamical systems based on pseudo-leader mechanism,'' Syst. Contr. Lett., vol. 61, no. 1, pp. 195-202, 2012.




\end{thebibliography}

\end{document}